\documentclass[a4paper,11pt]{amsart}
\usepackage[utf8]{inputenc}
\usepackage[T1]{fontenc}
\usepackage{enumerate}
\usepackage[normalem]{ulem}
\usepackage{cancel}
\usepackage{amsmath}
\usepackage{amssymb}
\usepackage{amsfonts}
\usepackage{mathrsfs}

\usepackage{hyperref}

\usepackage{amsthm}
\title[Gap opening for Dirac pseudo-differential operators]{Gap opening in the spectrum of some Dirac-like pseudo-differential operators}
\author[J.-M. Barbaroux]{J.-M. Barbaroux}
\address[J.-M. Barbaroux]{Aix Marseille Univ, Universit\'e de Toulon, CNRS, CPT, Marseille, France}
\email{jean-marie.barbaroux@univ-tln.fr}
\author[H.D. Cornean]{H.D. Cornean}
\address[H.D. Cornean]{Department of Mathematical Sciences, Aalborg University
\\
Skjernvej 4A, 9220 Aalborg \O, Denmark}
\email{cornean@math.aau.dk}
\author[S. Zalczer]{S. Zalczer}
\address[S. Zalczer]{Aix Marseille Univ, Universit\'e de Toulon, CNRS, CPT, Marseille, France}
\email{sylvain.zalczer@univ-tln.fr}

\newcommand\intdir[3]{\int_{#1}^\oplus #3 #2}
\newtheoremstyle{RRMPA}%
{}{}%
{\itshape}
{}
{\textsc}
{.}
{ }
{}
\newtheoremstyle{RRMPAd}%
{}{}%
{}
{}
{\itshape}
{.}
{ }
{}
\theoremstyle{RRMPA}
\newtheorem{prop}{PROPOSITION}
\newtheorem{lemme}{LEMMA}
\newtheorem{theorem}{THEOREM}

\theoremstyle{RRMPAd}

\newtheorem{remark}[theorem]{Remark}
\newtheorem{hypothesis}{Hypothesis}

\newcommand{\Dom}{\mathrm{Dom}}
\newcommand{\C}{\mathbb{C}} % Complex numbers
\newcommand{\Z}{\mathbb{Z}} % Integer numbers
\newcommand{\R}{\mathbb{R}} % Real numbers

\newcommand{\Ran}{\mathrm{Ran}}
\newcommand{\spec}{\mathrm{Spec}}

\newcommand{\sps}[2]{\langle #1,#2 \rangle} %Scalar Product
\newcommand{\norm}[1]{\left\lVert #1 \right\rVert}

\newcommand{\bgamma}{\boldsymbol{\gamma}}
\newcommand{\bs}{\boldsymbol{\sigma}}
\newcommand{\bx}{\mathbf{x}}
\newcommand{\by}{\mathbf{y}}

\newcommand{\bp}{\boldsymbol{p}}
\newcommand{\bchi}{\boldsymbol{\chi}}
\newcommand{\bP}{\boldsymbol{\Phi}}
\newcommand{\bF}{\boldsymbol{F}}
\newcommand{\bG}{G}
\newcommand{\bz}{\mathbf{0}}
\newcommand{\Id}{\mathrm{Id}}
\newcommand{\bbf}{\boldsymbol{f}}
\newcommand{\bg}{\boldsymbol{g}}
\newcommand{\bK}{K}
\newcommand{\bU}{\boldsymbol{U}}
\newcommand{\bW}{\boldsymbol{W}}
\newcommand{\bPsi}{\boldsymbol{\Psi}}
\newcommand{\bUpsilon}{\boldsymbol{\Upsilon}}
\newcommand{\bpsi}{\boldsymbol{\psi}}

\newcommand{\bk}{\mathbf{k}}
\newcommand{\bm}{\mathbf{m}}

\newcommand{\fesh}{\mathcal{F}}
\newcommand{\B}{\mathcal{B}}
\renewcommand{\le}{\leqslant}
\renewcommand{\leq}{\leqslant}
 \renewcommand{\ge}{\geqslant}
 \renewcommand{\geq}{\geqslant}

\begin{document}

\begin{abstract}\scriptsize
 In this paper, we study the opening of a spectral gap for a class of 2-dimensional periodic Hamiltonians which include those modelling multilayer graphene. 
 The kinetic part of the Hamiltonian is given by $\bs \cdot \bF(-i\nabla)$, where $\bs$ denotes the Pauli matrices and $\bF$ is a sufficiently regular vector-valued function which equals 0 at the origin and grows at infinity. 
 Its spectrum is the whole real line. We prove that a gap appears for perturbations in a certain class of periodic matrix-valued potentials depending on $\bF$, and we study how this gap depends on different parameters.
%A gap near 0 appears after perturbing with a certain periodic matrix-valued potential, and we study how this gap depends on different parameters. 
\end{abstract}
\maketitle

%%%%%%%%%%%%%%%%%%%%%%%%%%%%%%%%%%%%%%%%%%%%%%%%%%%
%%%%%%%%%%%%%%%%%%%%%%%%%%%%%%%%%%%%%%%%%%%%%%%%%%%
%%%%%%%%%%%%%%%%%%%%%%%%%%%%%%%%%%%%%%%%%%%%%%%%%%%
\section{Introduction, model and main result}

Graphene is a two-dimensional material made of carbon atoms on a honeycomb lattice. Among its remarkable properties is its energy band structure, with two bands crossing at the Fermi level \cite{CGPNGG}. This particular structure has suggested to model the dynamics of one electron in a graphene sheet by the free massless two-dimensional Dirac operator.

An interesting problem is to study the electronic properties of a material which is not a single sheet of graphene but several stacked layers of graphene. In this case, the dynamics of the electron can be approximated by an effective Hamiltonian which typically is a $N$th order Dirac-like operator, 
$N$ being the number of layers (see \cite{MCK} and references therein).

One of the major problems linked with graphene is to tune an energy bandgap at the Fermi level, making graphene a semiconductor. To realize this, one of the possibilities is to use the so-called graphene antidot lattices, which consist of a sheet of graphene periodically patterned with obstacles such as holes. In the case of single-layer graphene antidot lattices, the gap opening has been  numerically achieved in \cite{FPFMBPJ} and was proved in \cite{barbaroux} with a mathematical approach. See also \cite{BCLS, BCZ} and references therein for rigorous studies of spectral properties of Dirac operators modelling graphene antidot lattices.   

The goal of this article is to generalize such gap opening results to higher-order Hamiltonians, including the ones for multilayer graphene.

Namely, we want to study gap opening in the spectrum under periodic pertubations of the Hamiltonian  
 $$
  H_0=\bs\cdot \bF(-i\nabla)
 $$
on $L^ 2(\mathbb{R}^ 2,\mathbb{C}^ 2)$,  where $\bs=( \sigma_1,\,\sigma_2,\,\sigma_3)$ denotes the usual Pauli matrices  
\begin{align*}
 \sigma_1
  =\left(
    \begin{array}{cc}
      0&1\\
      1&0
    \end{array}
   \right),\quad
 \sigma_2=\left(
  \begin{array}{cc}
   0&-i\\
   i&0
  \end{array}
 \right),\quad 
 \sigma_3
 =\left(
  \begin{array}{cc}
    1&0\\
    0&-1
\end{array}
 \right),
\end{align*} 
and $\bF:\mathbb{R}^2\rightarrow \mathbb{R}^3$.  More precisely, this means that for any $\bpsi\in\Dom(H_0)$ 
 \[
  H_0\bpsi(\bx)= \left(\mathcal{F}^{-1}
  \left[ \bs\cdot \bF(\cdot) \mathcal{F}\bpsi\right]\right)(\bx), 
 \]
where $\mathcal{F}$ denotes the Fourier transform on $L^2(\R^2,\C^2)$ 
          \[\mathcal{F}\bpsi(\bp)=\frac{1}{2\pi}\int_{\mathbb{R}^2}e^{-i\bp\cdot \bx}\bpsi(\bx)d\bx.\]

We suppose that the function $\bF$ fulfils the following assumptions.
%%%%%%%%%%%%%%%%  HYPOTHESIS ON F  %%%%%%%%%%%%%%%%
\begin{hypothesis}\label{hyp1}
\begin{enumerate}[(i)]
           \item $\bF$ belongs to $\mathcal{C}^3(\R^2,\, \R^3)$.
           \item There exist constants $K_0'$, $K_i>0$ such that for all $\bp\in\mathbb{R}^2$,
           \begin{equation}\label{F}
           \begin{split}
           K_0'|\bp|^d\leqslant|\bF(\bp)|\leqslant K_0|\bp|^d \, ,\\
           |D^i\bF(\bp)|\leqslant K_i<\bp>^{d-|i|}
		   \end{split}           
           \end{equation} 
for some $d>0$ and any multi-index $i$ such that $1\leqslant |i|\leqslant 3$. Here $<\bp>=\sqrt{1+|\bp|^2}$ and $D^i$ denotes the multi-index partial derivative operator.
           \item There exists a $2\times 3$ rank 2 matrix $A$ such that in a neighbourhood of 0, 
           $$
           \bF(\bp)=|\bp|^{d-1}A\bp +O(|\bp|^{d+1}) .
           $$         
\end{enumerate}
\end{hypothesis}          
%%%%%%%%%%%%%%%%%%%%%%%%%%%%%%%%%%%%%%%%%%%%%%%%%%%

%%%%%%%%%%%%%%%%%%%%%%%%%%%%%%%%%%%%%%%%%%%%%%%%%%%   
\begin{prop}
The operator $H_0$ is unitarily equivalent  to a multiplication operator by  
$\sigma_3 |\bF|$. It is then self-adjoint on $\Dom(H_0)=\mathcal{F}^{-1}(\Dom(|\bF(\cdot)|))$ and its spectrum is given by  the essential range of $\pm|\bF|$, which is $\mathbb{R}$ under Hypothesis~\ref{hyp1} (i) and (ii).
\end{prop} 
%%%%%%%%%%%%%%%%%%%%%%%%%%%%%%%%%%%%%%%%%%%%%%%%%%%

%%%%%%%%%%%%%%%%%%%%%%%%%%%%%%%%%%%%%%%%%%%%%%%%%%%
\begin{proof}
The proof, identical to the one for the free Dirac operator, comes directly from the definition of $H_0$ through the Fourier transform which is unitary (cf. \cite{Thaller}).
\end{proof}
%%%%%%%%%%%%%%%%%%%%%%%%%%%%%%%%%%%%%%%%%%%%%%%%%%%

In order to open a gap around the zero energy, we will perturb $H_0$ with a periodic potential defined as follows. Let
\[\bchi=(\chi_1,\chi_2,\chi_3):\mathbb{R}^2\to\mathbb{R}^3\] 
where each $\chi_i$ is a bounded function with compact support included in the set $\Omega=]-\frac{1}{2},\frac{1}{2}]^2$.

Let $\beta>0$ and $\alpha\in]0,1]$.
The perturbed Hamiltonian is       
\begin{equation*} 
 H(\alpha,\beta)=H_0+\beta \sum_{\bgamma \in \mathbb{Z}^2}\bchi
 \left(\frac{\bx-  \bgamma }{\alpha }\right)\cdot\bs.
\end{equation*}

The operator $H(\alpha, \beta)$ is $\mathbb{Z}^2$-periodic and self-adjoint on $\Dom(H_0)$.

%%%%%%%%%%%%%%%%%%%%%%%%%%%%%%%%%%%%%%%%%%%%%%%%%%%
\begin{remark}
In \cite{barbaroux}, the authors treated the particular case corresponding to the free massless Dirac operator where $\bF(\bp)=(p_1,\,p_2,\,0)$ and $\chi_1=\chi_2=0$.
\end{remark}
%%%%%%%%%%%%%%%%%%%%%%%%%%%%%%%%%%%%%%%%%%%%%%%%%%%

Let us denote $\Phi_i=\int_\Omega\chi_i(\bx)d\bx$, $1\leq i\leq 3,$ and let us introduce the three-dimensional vector 
\[
\bP=(\Phi_i)_{1\leqslant i\leqslant3}\in \mathbb{R}^3.
\] 
We denote by $\bP_{\vert \vert}$ the projection of $\bP$ on $\Ran(A)$ and by $\Phi_\perp$ the projection on $\Ran(A)^\perp$.

%%%%%%%%%%%%%%%%   MAIN THEOREM   %%%%%%%%%%%%%%%%%
Here is the main result of our paper.
\begin{theorem}\label{thm:main}

Suppose that $\Phi_\perp\neq0$. Let $d'=\min(d,2)$.
There exist some constants $\lambda_0,C>0$ and $\delta \in ]0,1[$ with $C\delta<\frac{|\Phi_\perp|}{2}$ such that for any $\alpha \in  ]0,1/2]$ and $\beta>0$ satisfying $\alpha^2\beta<\lambda_0$, $\alpha^{d'}\beta<\delta$ we have that the interval
\[ \left[-\alpha^2\beta\left(\frac{|\Phi_\perp|}{2}-C\alpha^{d'}\beta\right), \alpha^2\beta\left(\frac{|\Phi_\perp|}{2}-C\alpha^{d'}\beta\right)\right]\]
belongs to the resolvent set $\rho(H(\alpha,\beta))$. 
\end{theorem}

\begin{remark}In \cite{barbaroux}, this condition is achieved since the kinetic part is in the subspace spanned by $\sigma_1$ and $\sigma_2$ and the potential is in $\mathrm{Span}(\sigma_3)$.
\end{remark}

As in \cite{barbaroux}, we use the Floquet-Bloch transformation to come to a problem on the unit square, where the gradient has a well-known eigenbasis.
Then, we use a Feshbach map argument, separating the problem between constant and non-constant modes. While the estimate on the constant subspace is direct, we need to use decay of the resolvent of the free operator and repeated applications of the resolvent equation to prove the invertibility on the orthogonal.

%\textcolor{blue}{\sout{As in \cite{barbaroux}, the proof is based on a Feshbach map argument. }}
The paper is organized as follows. In Section 2 we perform a detailed analysis of the integral kernel of the free resolvent, including its local singularities and off-diagonal decay.
In Section~3 we give the proof of the main theorem, while in the Appendix we summarize the results we need from the Bloch-Floquet transformation.

%%%%%%%%%%%%%%%%%%%%%%%%%%%%%%%%%%%%%%%%%%%%%%%%%%%
%%%%%%%%%%%%%%%%%%%%%%%%%%%%%%%%%%%%%%%%%%%%%%%%%%%
%%%%%%%%%%%%%%%%%%%%%%%%%%%%%%%%%%%%%%%%%%%%%%%%%%%
\section{Resolvent decay and integral kernel}
%\textcolor{blue}{\sout{Due to the more complicated expression of the kinetic energy, we do not have an explicit formula for the integral kernel of the free resolvent. \textcolor{blue}{In this chapter, }we prove that the integral kernel exists, it has integrable local singularities and has a sufficiently fast off-diagonal polynomial decay.}}

In this chapter, we study the behaviour of the integral kernel for the free resolvent. Due to the expression of the kinetic energy, we do not have an explicit formula for this integral kernel. Nevertheless, we can  prove that the integral kernel exists, it has integrable local singularities and has a sufficiently fast off-diagonal polynomial decay.

The result states as follows.
%%%%%%%%%%%%%%%%%%%%%%%%%%%%%%%%%%%%%%%%%%%%%%%%%%%

%%%%%%%%%%%%%%%%%%%%%%%%%%%%%%%%%%%%%%%%%%%%%%%%%%%
%From the proof of Lemma~\ref{reso}, we can infer the following result.
\begin{prop}\label{kernel-existence}
Let $M_d$ be the function defined on $\mathbb{R}^2\times\mathbb{R}^2$ by:
\begin{align*} 
 M_d(\bx,\bx')=\left\{\begin{matrix}\frac{1}{|\bx-\bx'|^{2-d}}+1 &\text{if}& |\bx-\bx'|\leqslant 1\text{ and }d\neq 2 ; \\
                      -\log|\bx-\bx'|+1&\text{if}&|\bx-\bx'|\leqslant 1\text{ and }d=2 ; \\
                      \frac{1}{|\bx-\bx'|^3} &\text{if}& |\bx-\bx'|\geqslant 1 .\end{matrix}\right.
                      \end{align*}
                      The operator $(H_0-i)^{-1}$ has an integral kernel, denoted by $(H_0-i)^{-1}(\bx,\bx')$, such that 
\begin{equation}\label{majNI}
 \left| (H_0-i)^{-1}(\bx,\bx') \right| \leqslant M_d(\bx,\bx')\, .
\end{equation}
\end{prop}
%%%%%%%%%%%%%%%%%%%%%%%%%%%%%%%%%%%%%%%%%%%%%%%%%%%

The proof of this Proposition is based on the following Lemma, and is postponed to the end of this section.

%%%%%%%%%%%%%%%%    LEMME 1 Section 2 %%%%%%%%%%%%%
\begin{lemme}\label{reso}
There exists some $C>0$ such that for all $\bbf$ and $\bg\in L^2(\mathbb{R}^2,\mathbb{C}^2)$ we have
\begin{equation}\label{majoreso}|\langle \bbf,(H_0\pm i)^{-1}\bg\rangle|\leqslant C\int_{\mathbb{R}^2}\int_{\mathbb{R}^2}|\bbf(\bx)|M_d(\bx,\bx')|\bg(\bx')|d\bx d\bx'.\end{equation}
\end{lemme}
%%%%%%%%%%%%%%%%%%%%%%%%%%%%%%%%%%%%%%%%%%%%%%%%%%%

%%%%%%%%%%%%%%%%%%%%%%%%%%%%%%%%%%%%%%%%%%%%%%%%%%%
\begin{remark}
   If $d> 2$ then $M_d$ is bounded while if $d<2$, $M_d$ has a local singularity of the type  $\frac{1}{|\bx-\bx'|^{2-d}}$. 
   This is why we denote $d'=\min(d,2)$; we can then write if $d\neq 2$ $M_d\leq\frac{2}{|\bx-\bx'|^{2-d'}}$.
\end{remark}
%%%%%%%%%%%%%%%%%%%%%%%%%%%%%%%%%%%%%%%%%%%%%%%%%%%

%%%%%%%%%%%%%%%%%%%%%%%%%%%%%%%%%%%%%%%%%%%%%%%%%%%
\begin{proof}[Proof of Lemma~\ref{reso}]
For $\bp\in \R^2$, we define the $2\times 2$ matrix 
$$
 \bG(\bp)=(\bs\cdot \bF(\bp)-i)^{-1}.
$$ 
The operator of multiplication by $\bG$ is bounded on $L^2(\R^2,\, \C^2)$.    
 
Remind that $<\bp>=\sqrt{1+|\bp|^2}$.
We define, for $\epsilon >0$, the regularized kernel
 \begin{align*}%\label{hc1}
  \bK_\epsilon(\bx,\bx')= \frac{1}{2\pi} \int_{\mathbb{R}^2}e^{i\bp\cdot(\bx-\bx')}
  \bG(\bp)e^{-\epsilon<\bp>}d\bp = \mathcal{F}^{-1}
  (\bG e^{-\epsilon <\cdot>})(\bx-\bx').
 \end{align*}
 
The estimates \eqref{F} of Hypothesis~\ref{hyp1}(ii) implies that for any multi-index $N$ such that $|N|\leqslant 3$ there exists $C_N>0$ such that
\begin{equation}\label{eq:est-G(p)}
 |D^N \bG(p)|\leqslant\frac{C_N}{<\bp>^{|N|+d}}.
\end{equation}
For $\bx=(x_1,x_2)$ and $\bx'=(x'_1, x'_2)$, choose $l\in\{1,2\}$ such that  $|\bx-\bx'|\leq\sqrt{2}|x_l-x'_l|$. By repeated integrations by part, we find that for any integer $M\leq 3$ we have :
 \begin{equation}\label{K_epsbis}
 (x_l-x_l')^M \bK_\epsilon(\bx,\bx')=i^{M}\frac{1}{2\pi}\int_{\mathbb{R}^2}e^{i\bp\cdot(\bx-\bx')}\frac{\partial^M}{\partial p_l^M}\left(\bG(\bp)e^{-\epsilon<\bp>}\right)d\bp .
 \end{equation}
Hence we have 
\begin{equation}\label{K_eps}
 |x_l-x_l'|^M |\bK_\epsilon(\bx,\bx')|\leqslant\frac{1}{2\pi}
 \int_{\mathbb{R}^2}\left|\frac{\partial^M}{\partial p_l^M}
 \left(\bG(\bp)e^{-\epsilon<\bp>}\right)\right|d\bp.
\end{equation}
Pick $M=3$. We have then, by product rule and denoting $E(\bp)=e^{-\epsilon<\bp>}$, 
 \begin{equation*}%\label{eq:est-1}
 \begin{split}
 \frac{\partial^3}{\partial p_l^3}(\bG(\bp)E(\bp))=&
 \frac{\partial^3\bG(\bp)}{\partial p_l^3}E(\bp)
 +3\frac{\partial^2\bG(\bp)}{\partial p_l^2}\frac{\partial E(\bp)}{\partial p_l}\\
 &+3\frac{\partial \bG(\bp)}{\partial p_l}
 \frac{\partial E(\bp)^2}{\partial p_l^2}
 +\bG(\bp)\frac{\partial E(\bp)^3}{\partial p_l^3}.
 \end{split}
 \end{equation*}
Furthermore, we have: 
\begin{equation}\label{eq:est-2}
\begin{split}
 \frac{\partial E(\bp)}{\partial p_l} 
  = & -\epsilon p_l\frac{E(\bp)}{<\bp>} \, ,\\
 \frac{\partial^2 E(\bp)}{\partial p_l^2}
  = & -\epsilon\frac{E(\bp)}{<\bp>}
 +\epsilon p_l^2\frac{E(\bp)}{<\bp>^3}+\epsilon^2p_l^2\frac{E(\bp)}{<\bp>^2} \, ,\\
 \frac{\partial^3 E(\bp)}{\partial p_l^3}
 = & \, 3\epsilon p_l\frac{E(\bp)}{<\bp>^3}+3\epsilon^2 p_l
 \frac{E(\bp)}{<\bp>^2}-3\epsilon p_l^3\frac{E(\bp)}{<\bp>^5}\\
 & -3\epsilon^2 p_l^3\frac{E(\bp)}{<\bp>^4}-\epsilon^3 p_l^3\frac{E(\bp)}{<\bp>^3}\, .
\end{split}
\end{equation}
%Therefore, using \eqref{eq:est-1} and \eqref{eq:est-2}, we  bound from above the right hand side of \eqref{K_eps} by a sum of several terms that we will estimate using the fact that $|E(\bp)|\leqslant 1$ and $\frac{|p_l|}{<\bp>}\leqslant 1$.
Moreover, since $x^ke^{-x}$ is bounded on $\R^+$ for all $k$, there exist some constants $c_k$ such that for all $\epsilon>0$ and $\bp\in\R^2$
\begin{equation}\label{eq:trivial1}
\epsilon^ke^{-\epsilon<\bp>}\leqslant c_k<\bp>^{-k}.
\end{equation}
In the sequel, denoting by $C$ a generic constant \emph{independent of $\epsilon$}, we obtain from \eqref{eq:est-2}, the above bound \eqref{eq:trivial1}, and the fact that $|E(\bp)|\leqslant 1$ and $\frac{|p_l|}{<\bp>}\leqslant 1$, that for $j\in\{1,2,3\}$,
%From the above bound, we find that for $j\in\{1,2,3\}$, 
\begin{equation}\label{eq:ast-ast-ast}
 \left|\frac{\partial^j E(\bp)}{\partial p_l^j}\right|\leqslant\frac{C}{<\bp>^j}\, .
\end{equation}
% \begin{itemize}
%  \item $\left|\frac{\partial E(p)}{\partial p_l}\right|\leqslant\frac{C}{<p>}$ ;
%  \item $\left|\frac{\partial^2 E(p)}{\partial p_l^2}\right|\leqslant\frac{C}{<p>^2}$ ;
%  \item $\left|\frac{\partial^3 E(p)}{\partial p_l^3}\right|\leqslant\frac{C}{<p>^3}$.
%\end{itemize}
% We will use these estimates and the ones on the derivatives of $G$ to estimate the  integrands.
From \eqref{eq:est-G(p)} and \eqref{eq:ast-ast-ast} we obtain for 
$j\in \{0,1,2,3 \}$,
\begin{equation*}%\label{eq:est-power3}
 \int_{\R^2} \left| \frac{\partial^j}{\partial p_j{}^j} G(\bp) \right| \, 
 \left|   \frac{\partial^{3-j}}{\partial p_j{}^{3-j} }E(\bp)\right| d\bp 
 \leq C
\end{equation*}
%Henceforth,  we get:
%\begin{equation}\label{eq:est-power3}
% \int_{\mathbb{R}^2}\left|\frac{\partial^3\bG(\bp)}{\partial p_l^3}\right|
% |E(\bp)|d\bp
% \leqslant\int_{\mathbb{R}^2}\frac{C}{<\bp>^{3+d}}d\bp\leqslant C.
%\end{equation}
%
%\[\int_{\mathbb{R}^2}\left|\frac{\partial^2\bG(\bp)}{\partial p_l^2}\right|\left|\frac{\partial E(\bp)}{\partial p_l}\right|d\bp\leqslant\int_{\mathbb{R}^2}\frac{C}{<\bp>^{2+d}<\bp>}d\bp\leqslant C.\]
%
%\[\int_{\mathbb{R}^2}\left|\frac{\partial \bG(\bp)}{\partial p_l}\right|\left|\frac{\partial^2 E(\bp)}{\partial p_l^2}\right|d\bp\int_{\mathbb{R}^2}\frac{C}{<\bp>^{1+d}<\bp>^2}d\bp\leqslant C.\]
%
%\[\int_{\mathbb{R}^2}|\bG(\bp)|\left|\frac{\partial^3 E(\bp)}{\partial p_l^3}\right|d\bp\leqslant\int_{\mathbb{R}^2}\frac{C}{<\bp>^{d}<\bp>^3}d\bp\leqslant C.\]
Hence, according to \eqref{K_eps}, we have $|\bK_\epsilon(\bx,\bx')|\leqslant \frac{C}{|x_l-x'_l|^3}$ and thus
\begin{equation}
 |\bK_\epsilon(\bx,\bx')|\leqslant \frac{C}{|\bx-\bx'|^3}. \label{majoX}
 \end{equation}
This estimate is only useful if $|\bx-\bx'|\geq 1$.

Let us now study the case $|\bx-\bx'|\leqslant 1$. We write
\begin{equation}\label{eq:est-K}
\begin{split}
 2\pi\bK_\epsilon(\bx,\bx')=&\int_{|\bp|\leqslant 1}e^{i \bp\cdot(\bx-\bx')}
 \bG(\bp)e^{-\epsilon<\bp>}d\bp\\
 &+\int_{1\leqslant|\bp|
 \leqslant |\bx-\bx'|^{-1}}e^{i\bp\cdot(\bx-\bx')}\bG(\bp)e^{-\epsilon<\bp>}d\bp\\&
 +\int_{|\bx-\bx'|^{-1} \leqslant |\bp|}e^{i\bp\cdot(\bx-\bx')}\bG(\bp)e^{-\epsilon<\bp>}d\bp.
\end{split}
\end{equation}
We simply bound the first term in the right hand side of \eqref{eq:est-K} by 
\begin{equation}\label{majo1}
 \left|\int_{|\bp|\leqslant 1}e^{i\bp\cdot(\bx-\bx')}
\bG(\bp)e^{-\epsilon<\bp>}d\bp\right| \leqslant \int_{|\bp|\leqslant 1}|\bG(\bp)|d\bp  \, ,
\end{equation}
which is finite and independent of $\epsilon$.

To bound the second term in the right hand side of \eqref{eq:est-K}, we use from Hypothesis~\ref{hyp1}(ii) that $|\bG(\bp)|\leqslant\frac{C}{|\bp|^d}$ which yields
\begin{equation}\label{majo2} 
 \left| \int_{1\leqslant |\bp|\leqslant |\bx - \bx'|^{-1}}e^{i \bp\cdot(\bx-\bx')}
 \bG(\bp)e^{-\epsilon<\bp>}d\bp \right|
 \leqslant
 2\pi \int_1^{|\bx-\bx'|^{-1}} \frac{C}{r^d}r dr
 %\int_{1\leqslant|\bp|
 %\leqslant |\bx-\bx'|^{-1}}|\bG(\bp)|d\bp  \leq 2\pi
 %\int_1^{|\bx-\bx'|^{-1}} \frac{C}{r^d}r dr=\frac{2\pi C}{2-d}
 %\left(\frac{1}{|\bx-\bx'|^{2-d}}-1\right)\, ,
\end{equation}
which is bounded by $C \left( \frac{1}{|\bx - \bx'|^{2-d}} -1\right)$
for $d\neq 2$ and by $C\log\left( |\bx-\bx'|^{-1}\right)$ for $d=2$.

To estimate the third term in the right hand side of \eqref{eq:est-K}, we need some more care.  We choose $l\in\{1,2\}$ as before such that $|\bx-\bx'|\leq \sqrt{2}|x_l-x_l'|$.

Let us calculate 
\[
 \int_{|\bp|\geqslant |\bx-\bx'|^{-1}}-(x_l-x'_l)^2e^{i\bp\cdot(\bx-\bx')}
 \bG(\bp)e^{-\epsilon<\bp>}d\bp \, ,
\] 
which corresponds to the integral that we want to estimate multiplied by $-(x_l-x'_l)^2$. For $\theta\in [0,2\pi)$ we define the vector $\bp(\theta)=(|\bx-\bx'|^{-1}\cos\theta,|\bx-\bx'|^{-1}\sin\theta)$ and we denote by $p_l(\theta)$ its $l$-th component. Then, integrating by part with respect to the $p_l$ variable and applying Gauss divergence theorem, we have
\begin{equation}\label{est-gauss}
\begin{split}
 &\int_{|\bp|\geqslant |\bx-\bx'|^{-1}}-(x_l-x'_l)^2e^{i\bp\cdot(\bx-\bx')}
 \bG(\bp)e^{-\epsilon<\bp>}d\bp\\
 & =-\int_0^{2\pi}e^{-\epsilon<\bp(\theta)>}i(x_l-x_l')
 e^{i\bp(\theta)\cdot(\bx-\bx')}\bG(\bp(\theta))p_l(\theta) d\theta\\
 &\ \ \ -\int_{|\bp|\geqslant |\bx-\bx'|^{-1}}i(x_l-x'_l)
 e^{i\bp\cdot(\bx-\bx')}\frac{\partial}{\partial p_l}
 \left(\bG(\bp)e^{-\epsilon<\bp>}\right)d\bp.
\end{split}
\end{equation}
Using the estimate \eqref{eq:est-G(p)}, we get that the first term is bounded by
\begin{equation}\label{majo31}
 2\pi\frac{C_0}{(|\bx-\bx'|^{-1})^d}|\bx-\bx'|^{-1}|x_l-x'_l|
 \leqslant 2\pi C_0|\bx-\bx'|^d.
\end{equation}

To estimate the second term in the right hand side of \eqref{est-gauss}, we use a new integration by parts:
\begin{equation}\label{eq:added-1}
\begin{split}
 & -\int_{|\bp|\geqslant |\bx-\bx'|^{-1}}i(x_l-x'_l)
 e^{i\bp\cdot(\bx-\bx')}\frac{\partial}{\partial p_l}\left
 (\bG(\bp)e^{-\epsilon<\bp>}\right)d\bp\\
  &=\int_0^{2\pi}e^{i\bp(\theta)\cdot(\bx-\bx')}\frac{\partial}{\partial p_l} 
  \left(\bG e^{-\epsilon<\cdot>}\right)(\bp(\theta))p_l(\theta) d\theta\\
  &\ \ \ +\int_{|\bp|\geqslant |\bx-\bx'|^{-1}}e^{i\bp\cdot(\bx-\bx')}
  \frac{\partial^2}{\partial p_l^2}\left(\bG(\bp)e^{-\epsilon<\bp>}\right)d\bp.
 \end{split}
\end{equation}
Using again that
\[\left|\frac{\partial}{\partial p_l}\left(\bG(\bp)e^{-\epsilon<\bp>}\right)\right|\leqslant\frac{C}{|\bp|^{d+1}}\text{ and } \left|\frac{\partial^2}{\partial p_l^2}\left(\bG(\bp)e^{-\epsilon<\bp>}\right)\right|\leqslant\frac{C}{|\bp|^{d+2}}\ ,\]
we can bound the first term in the right hand side of \eqref{eq:added-1} by
\begin{equation}\label{majo32}
 2\pi |\bx-\bx'|^{-1}\frac{C}{|\bx-\bx'|^{-(d+1)}}= 2\pi C|\bx-\bx'|^d
\end{equation}
and the second one by
\begin{equation}\label{majo33}
\int_{|\bp|\geqslant |\bx-\bx'|^{-1}}\frac{C}{|\bp|^{d+2}}d\bp
= \frac{2\pi C}{d} |\bx -\bx'|^d
%2\pi\int_{|\bx-\bx'|^{-1}}^{+\infty}\frac{C}{r^{d+2}}rdr
%=\frac{C}{|\bx-\bx'|^{-d}}=C|\bx-\bx'|^d.
\end{equation}

Putting together the estimates \eqref{majo31}, \eqref{majo32} and \eqref{majo33}, we find that
%\begin{equation}
% \left|\int_{|p|\geqslant |\bx-\bx'|^{-1}}-(x_l-x'_l)^2e^{ip\cdot(\bx-\bx')}G(p)e^{-%\epsilon<p>}dp\right|\leq C|\bx-\bx'|^d
%\end{equation}
%and then, by dividing with $|x_l-x_l'|^2$, we obtain
\begin{equation}\label{majo3}
 \left|\int_{|\bp|\geqslant |\bx-\bx'|^{-1}}e^{i\bp\cdot(\bx-\bx')}
 \bG(\bp)e^{-\epsilon<\bp>}d\bp\right|\leq C|\bx-\bx'|^{d-2}.
\end{equation}

Adding the estimates \eqref{majo1}, \eqref{majo2} and \eqref{majo3}, we find that there exists a constant $C>0$ such that uniformly in $\epsilon$ 
\[
 |\bK_\epsilon(\bx,\bx')|\leqslant C\left(\frac{1}{|\bx-\bx'|^{2-d}}+1\right)\, ,
\] 
if $|\bx-\bx'|\leqslant 1$, which together with the result \eqref{majoX} which holds  for  $|\bx-\bx'|\geqslant 1$ allows us to conclude that uniformly in 
$\epsilon$ we have
\[
 |\bK_\epsilon(\bx,\bx')|\leqslant CM_d(\bx,\bx')\, ,
\]
where $M_d$ is the function defined in the statement of the Lemma.

We are now ready to prove  the estimate \eqref{majoreso} for $(H_0-i)^{-1}$. Let $\bbf$ and $\bg$ be in the Schwartz space $\mathscr{S}(\mathbb{R}^2,\C^2)$. Then,
\[\langle \bbf,(H_0-i)^{-1}\bg \rangle =\int_{\mathbb{R}^2}\overline{\hat{\bbf}(\bp)}\bG(\bp)\hat{\bg}(\bp)d\bp\] by  Parseval's identity. By dominated convergence, we have 
\begin{equation}\label{CD}
 \int_{\mathbb{R}^2}\overline{\hat{\bbf}(\bp)}\bG(\bp)\hat{\bg}(\bp)d\bp
 =\lim_{\epsilon\to0^+}\int_{\mathbb{R}^2}
 \overline{\hat{\bbf}(\bp)}\bG(\bp)e^{-\epsilon<\bp>}\hat{\bg}(\bp)d\bp
 \end{equation} 
and, by Parseval's identity again and denoting by $*$ the convolution product between $L^1(\R^2, \mathcal{M}_2(\C))$ and $L^2(\R^2,\C^2)$,
\begin{align*}
 \int_{\mathbb{R}^2}\overline{\hat{\bbf}(\bp)}\bG(\bp)e^{-\epsilon<\bp>}\hat{\bg}(\bp)d\bp
 =&\int_{\mathbb{R}^2}\overline{\bbf(\bx)}(\bK_\epsilon(\cdot,0)*\bg)(\bx)d\bx\\
 =&\int_{\mathbb{R}^2}\int_{\mathbb{R}^2}
 \overline{\bbf(\bx)}\bK_\epsilon(\bx,\bx')\bg(\bx')d\bx d\bx'. 
\end{align*}

Knowing that 
\[
 \forall \epsilon>0,\  |\bK_\epsilon(\bx,\bx')|\leqslant M_d(\bx,\bx')\, ,
\]
we get, using \eqref{CD}, that
\[
 |\langle \bbf,(H_0-i)^{-1}\bg\rangle|\leqslant
 \int_{\mathbb{R}^2}|\bbf(\bx)|M_d(\bx,\bx')|\bg(\bx')|d\bx d\bx'\, .
\]

This concludes the proof of Lemma~\ref{reso}.
\end{proof}
%%%%%%%%%  end proof of Lemma section 2 %%%%%%%%%%%

We are now ready to give the proof of Proposition~\ref{kernel-existence}.
%%%%%%%%%%%%%%%%%%%%%%%%%%%%%%%%%%%%%%%%%%%%%%%%%%%
\begin{proof}[Proof of Proposition~\ref{kernel-existence}]
%Let us prove the existence of an integral kernel for $(H_0-i)^{-1}$. 
By equation \eqref{K_epsbis}, we know that for $\bx\neq \bx'$,
 
\[ 
 \bK_\epsilon(\bx,\bx')
 =\frac{i}{(x_l-x_l')^3}\frac{1}{2\pi}
 \int_{\mathbb{R}^2}e^{i\bp\cdot(\bx-\bx')}
 \frac{\partial^3}{\partial p_l^3}\left(\bG(\bp)e^{-\epsilon<\bp>}\right)d\bp.
\]
The integrand converges pointwise to $e^{i\bp\cdot(\bx-\bx')}\frac{\partial^3\bG(\bp)}{\partial p_l^3}$. Moreover, using \eqref{eq:est-G(p)}, it is dominated by some integrable function independent of $\epsilon$. 
Then, by dominated convergence, $\bK_\epsilon$ converges pointwise to some function of $\bx$ and $\bx'$ which will be denoted by $(H_0-i)^{-1}(\cdot,\cdot)$ and which trivially satisfies Inequality~\eqref{majNI}.
%We have moreover that $|K_0(\bx,\bx')|\leq M_d(\bx,\bx')$.
 
Then, by dominated convergence, we have that for all $\bbf$ and $\bg\in L^2(\R^2,\C^2)$
\begin{equation*}%\label{noyint}
 \langle \bbf,(H_0- i)^{-1}\bg\rangle
 =\int_{\mathbb{R}^2}\int_{\mathbb{R}^2}\bbf(\bx)
 (H_0-i)^{-1}(\bx,\bx')\bg(\bx')d\bx d\bx'\, ,
\end{equation*}
and
$$
|(H_0-i)^{-1}(\bx,\bx')|\leq M_d(\bx,\bx').
$$
\end{proof}
%%%%%%%%%%%%%%%%%%%%%%%%%%%%%%%%%%%%%%%%%%%%%%%%%%%

%%%%%%%%%%%%%%%%%%%%%%%%%%%%%%%%%%%%%%%%%%%%%%%%%%%
%%%%%%%%%%%%%%%%%%%%%%%%%%%%%%%%%%%%%%%%%%%%%%%%%%%
%%%%%%%%%%%%%%%%%%%%%%%%%%%%%%%%%%%%%%%%%%%%%%%%%%%
\section{Proof of the main theorem}

Let $\mathscr{S}(\R^2, \C)$ be the Schwartz space of test functions, and let us fix $\Omega= (-\frac12, \frac12]^2$. We define the Bloch-Floquet transformation by the map  
\begin{equation}\label{defbloch}
 \begin{split}
 \mathcal{U}:\quad& \mathscr{S}(\mathbb{R}^2,\, \C)\subset L^2(\mathbb{R}^2,\, \C)\to L^2(\Omega^2) \\
 & (\mathcal{U}\psi)(\bx,\bk)
 =\sum_{\bgamma\in\mathbb{Z}^2}e^{2i\pi \bk\cdot(\bx+\bgamma)}\psi(\bx+\bgamma)\, ,
 \end{split}
\end{equation} 
extended by density to $L^2(\mathbb{R}^2, \C)$.
It is possible to show (cf. \cite{bloch}) that $\mathcal{U}$ is a unitary operator and that for $f\in L^2(\Omega^2)$, $\bx\in\Omega$ and $\bgamma\in\mathbb{Z}^2$, 
\begin{equation*}%\label{defbloch*}
 (\mathcal{U}^*f)(\bx+\bgamma)
 =\int_\Omega e^{-2i\pi \bk\cdot(\bx+\bgamma)}f(\bx,\bk)d\bk.
\end{equation*} 
We then define the Bloch-Floquet transformation componentwise on $L^2(\mathbb{R}^2,\mathbb{C}^2)$ which will be abusively again denoted by 
$\mathcal{U}$.

%On note également $r$ la multiplication d'une fonction de $L^2(\Omega\times \Omega)$ par $e^{2\pi i \bk\cdot \bx}$ ; c'est aussi un opérateur unitaire.
Applying this tranformation (see Proposition~\ref{appbloch} in the appendix), we find
\begin{equation}\label{bloch}
\begin{split}
 & \mathcal{U}H(\alpha,\beta)\mathcal{U}^*
  =\intdir{\Omega}{d\bk}{h_\bk(\alpha,\beta)}\, ,\\
 & h_\bk(\alpha,\beta)
  =\bs\cdot\bF(-i\nabla_{\mathrm{per}}-2\pi \bk)+\beta\bchi_\alpha(\bx)\cdot\bs,
\end{split}
\end{equation}
where for each $\bk$ the fiber Hamiltonian $h_\bk(\alpha,\beta)$ is an  operator defined on $L^2(\Omega,\mathbb{C}^2)$. Similarly, we wil denote \begin{equation*}h_\bk^{(0)}=\bs\cdot\bF(-i\nabla_{\mathrm{per}}-2\pi \bk).\label{hk0}\end{equation*}
The operator $\nabla_{\mathrm{per}}$ means here the gradient on $L^2(\Omega,\mathbb{C}^2)$ with periodic boundary conditions and $\bchi_\alpha(\bx)=(\chi_{1,\alpha}(\bx),\chi_{2,\alpha}(\bx),\chi_{3,\alpha}(\bx)):=\bchi(\bx/\alpha)$.

The spectra of $H(\alpha,\beta)$ and $h_\bk(\alpha,\beta)$ are related through  (see  \cite[Theorem~XIII.85]{RS4})
  \begin{align}\label{spectra}
    \spec (H(\alpha,\beta))
    =\overline{\bigcup_{\bk\in\Omega}
    \spec (h_\bk(\alpha,\beta))}.
  \end{align}

Picking $\Psi_\bm=e^{2i\pi \bm\cdot \bx}$ for $\bm\in\mathbb{Z}^2$ and $\bx\in \Omega$, we get that the family of vectors $\Psi_\bm$ is a basis of eigenvectors of $\bF(-i\nabla_{\mathrm{per}})$
 satisfying if we denote $\bF=(F_1,F_2,F_3)$ 
\[
 F_i(-i\nabla_{\mathrm{per}})\Psi_\bm=F_i(2\pi \bm)\Psi_\bm\text{ for }i=1,2,3.
\]

We then define, for $\bm\in\mathbb{Z}^2$, the projections 
$$
P_\bm=|\Psi_\bm\rangle \langle\Psi_\bm|\otimes 1_{\mathbb{C}^2} \quad \mbox{and} \quad Q_\bz=\Id-P_\bz \, .
$$ %here

We  will use the Feshbach map method (see for example Lemma 6.1 of \cite{nenciu}) to reduce the spectral problem to problems on $P_\bz L^2(\Omega,\mathbb{C}^2)$ and $Q_\bz L^2(\Omega,\mathbb{C}^2)$.
This method claims that $z\in\rho(h_\bk(\alpha,\beta))$ if $Q_\bz(h_\bk(\alpha,\beta)-z)Q_\bz$ is invertible on  $Q_\bz L^2(\Omega,\mathbb{C}^2)$ and  the operator $\mathcal{F}_{P_\bz}(z)$ defined on $P_\bz L^2(\Omega,\mathbb{C}^2)$ and given by 
\begin{equation}\label{eq:feshbach-1}
\mathcal{F}_{P_\bz}(z)=P_\bz(h_\bk(\alpha,\beta)-z)P_\bz-\beta^2P_\bz\bchi_\alpha\cdot\bs Q_\bz(Q_\bz(h_\bk(\alpha,\beta)-z)Q_\bz)^{-1}Q_\bz\bchi_\alpha\cdot\bs P_\bz
\end{equation}
is also invertible.
 
We will first prove that $P_\bz h_\bk(\alpha,\beta)P_\bz$ has a spectral gap of order $\alpha^2\beta$ near 0 and then that the second term in the right hand side of \eqref{eq:feshbach-1} is small enough not to close the gap provided that $z$ is in the interval given in the theorem.
To show the invertibility of $P_\bz(h_\bk(\alpha,\beta)-z)P_\bz$, we have to bound from below $|\bF(\bp)+\lambda \bP|$, where we remind that we have denoted $\bP=(\Phi_i)_{1\leqslant i\leqslant3}$,  $\bP_{\vert \vert}$ the projection of $\bP$ on $\Ran(A)$ and $\Phi_\perp$ its projection on $\Ran(A)^\perp$.

%%%%%%%%%%%%%%%%%%%%%%%%%%%%%%%%%%%%%%%%%%%%%%%%%%%
\begin{lemme}\label{lem2.1}
 Let $\alpha\in]0,1[$ and $\beta>0$. Then, for every $\bk\in \Omega$ and $\Psi\in P_\bz L^2(\Omega,\mathbb{C}^2)$, we have for $\alpha^2\beta$ small enough:
\[
  \|P_\bz h_\bk(\alpha,\beta)P_\bz\Psi\|\geqslant \frac{|\Phi_\perp|}{2}\alpha^2\beta\|\Psi\|.
\]
\end{lemme}
%%%%%%%%%%%%%%%%%%%%%%%%%%%%%%%%%%%%%%%%%%%%%%%%%%%

%%%%%%%%%%%%%%%%%%%%%%%%%%%%%%%%%%%%%%%%%%%%%%%%%%%
\begin{proof}
We have, for $\bk\in\Omega$ :
\[
 P_\bz h_\bk(\alpha,\beta)P_\bz
 =(\bs\cdot \bF(-2\pi \bk)+\alpha^2\beta \bP \cdot\bs)P_\bz.
\]
Let us denote $\lambda=\alpha^2\beta$. For $\Psi\in P_\bz L^2(\Omega,\mathbb{C}^2)$, 
 \begin{align*}
  \|P_\bz h_\bk(\alpha,\beta)P_\bz\Psi\|^2&=\|\bs\cdot(\bF(-2\pi \bk)
  +\lambda \bP)P_\bz\Psi\|^2\\
				    &\geqslant \inf_{\bp\in\R^2}|\bF(\bp)+\lambda \bP|^2 \|\Psi\|^2.
				    %&=C^2\lambda^2\|\Psi\|^2
 \end{align*}  
The  lower bound in the lemma would follow if we can prove the following statement: there exists  $\lambda_0>0$ such that 
\begin{equation*}%\label{eq:to-be-proved}
 \inf_{\bp\in\mathbb{R}^2}|\bF(\bp)+\lambda \bP|\geqslant \frac{|\Phi_\perp|}{2}\lambda\, ,
\end{equation*}
for $0\leqslant\lambda\leqslant \lambda_0$.

In order to prove this, pick $M$ such that $K_0'M^d-|\bP|=\frac{|\Phi_\perp|}{2}$ where  $K_0'$ is the constant appearing in the first inequality in \eqref{F} of Hypothesis~\ref{hyp1}(ii). 

For $|\bp|\geqslant  M\lambda^{1/d}$ we have by the first inequality in  \eqref{F}:
\begin{equation}\label{eq:tbp-1}
 |\bF(\bp)+\lambda \bP|\geqslant K_0'|\bp|^d-\lambda |\bP|\geqslant(K_0'M^d-|\bP|)\lambda\geq\frac{|\Phi_\perp|}{2}\lambda.
\end{equation}
 
For $|\bp|\leqslant  M\lambda^{1/d}$, by Hypothesis~\ref{hyp1}(iii), we have for some $K>0$ and $\lambda$ small enough:
 \begin{align*}
 |\bF(\bp)+\lambda \bP|&= \left||\bp|^{d-1}A\bp+\lambda \bP_{\vert \vert}
 +\lambda \Phi_\perp\right|+O(|\bp|^{d+1})\geqslant\lambda |\Phi_\perp|
 -K|\bp|^{d+1}\\
  &\geqslant\lambda |\Phi_\perp|-KM^{d+1}\lambda^{(d+1)/d}.
  %\geqslant\frac{\lambda}{2} |\Phi_\perp|
 \end{align*}
Define $\lambda_0$ such that 
$$ 
 |\Phi_\perp|-KM^{d+1}\lambda_0^{1/d}=\frac{|\Phi_\perp|}{2}\, .
$$ 
For $\lambda\leq \lambda_0$, the above estimate implies 
\begin{equation}\label{eq:tbp-2}
 |\bF(\bp)+\lambda \bP|\geq\frac{\lambda}{2}|\Phi_\perp|\, ,
\end{equation}
for all $\bp\in\R^2$. 

Equations~\eqref{eq:tbp-1} and \eqref{eq:tbp-2} together conclude the proof.
\end{proof}
%%%%%%%%%%%%%%%%%%%%%%%%%%%%%%%%%%%%%%%%%%%%%%%%%%%

The invertibility of $Q_\bz(h_\bk(\alpha,\beta)-z)Q_\bz$ on $Q_\bz L^2(\Omega,\mathbb{C}^2)$ will require more technicality. We begin with the following estimates.
 Recall that $d'=\min(d,2)$.

%%%%%%%%%%%%%%%%%%%%%%%%%%%%%%%%%%%%%%%%%%%%%%%%%%%
\begin{lemme}
There exists $C$ such that for $|\alpha|\leq\frac{1}{2}$ and $\forall \bk\in \Omega$ we have 
\begin{align}
 \|\sqrt{|\bchi_\alpha|}P_\bz\|&\leqslant \alpha ;\label{eq:int-1}\\
 \||\bchi_\alpha|^{1/2}(h_\bk^{(0)}- i)^{-1}|\bchi_\alpha|^{1/2}\|
 &\leqslant C\alpha^{d'}\label{eq:res-est-1} ;\\ 
 \||\bchi_\alpha|^{1/2}(h_\bk^{(0)}- i)^{-1}\|
 &\leqslant C\sqrt{\alpha^{d'}},\label{eq:res-est-2}
\end{align} where $h_\bk^ {(0)}$ has been defined in Equation~\eqref{hk0}.
\end{lemme}
%%%%%%%%%%%%%%%%%%%%%%%%%%%%%%%%%%%%%%%%%%%%%%%%%%% 

%%%%%%%%%%%%%%%%%%%%%%%%%%%%%%%%%%%%%%%%%%%%%%%%%%%
\begin{proof}
As in \cite{barbaroux}, in order to show \eqref{eq:int-1} we compute for
  $\bbf,\bg\in L^2(\Omega,\C^2)$:
\[
|\sps{\bbf}{\sqrt{|\bchi_\alpha|}P_\bz \bg}| 
\le |\sps{\bbf}{\sqrt{|\bchi_\alpha|} \Psi_0}  |\,|\sps{\Psi_0}{\bg}|
\le \|\bchi_\alpha\|_1^{1/2} \|\bbf\|_2\|\bg\|_2\le \alpha \|\bbf\|_2\|\bg\|_2.
\]
 
For the next two inequalities, we need some notation: given an integral operator $T$, we denote its integral kernel by $T(\bx,\bx')$. We now use the following identity proved in Proposition~\ref{appnoyint} in the appendix:
 \[(h_\bk^{(0)}-i)^{-1}(\bx,\bx')=\sum_{\bgamma\in\mathbb{Z}^2} 
 e^{2i\pi \bk\cdot(\bx+\bgamma-\bx')}(H_0-i)^{-1}(\bx+\bgamma,\bx').\]
 
In the following, we will denote by $C$ any constant independent of $\alpha$ and $\bk$. Assume first that $d\neq2$. Let $\bUpsilon$, $\bPsi\in L^2(\Omega, \mathbb{C}^2)$ with $\bUpsilon$ with support in $\Omega_\alpha=[-\frac{\alpha}{2},\frac{\alpha}{2})^2$. 
According to Lemma~\ref{reso} and Proposition~\ref{kernel-existence}, we have:
\begin{align*}
 |\langle\bUpsilon, (h_\bk^{(0)}-i)^{-1}\bPsi\rangle|\leqslant&\sum_{\bgamma\in\mathbb{Z}^2} \iint_{\Omega_\alpha\times\Omega}|\bUpsilon(\bx)||(H_0-i)^{-1}|(\bx+\bgamma,\bx')|\bPsi(\bx')|d\bx d\bx'\\
				  \leqslant&\sum_{\bgamma\neq0} \iint_{\Omega_\alpha\times\Omega}|\bUpsilon(\bx)|\frac{C}{|\bx+\bgamma-\bx'|^3}|\bPsi(\bx')|d\bx d\bx'\\
				  &+\iint_{\Omega_\alpha\times\Omega}|\bUpsilon(\bx)|\frac{C}{|\bx-\bx'|^{2-d'}}|\bPsi(\bx')|d\bx d\bx'.
 \end{align*}
% 
%  In order to bound the first term, we note that, if $\gamma\neq 0$ and  $\Phi$ has support in $\Omega_\alpha$, $|\bx+\gamma-\bx'|\geqslant 1-\alpha$.
% Hence, we get, for any $\delta<1$:
%  \[\frac{1}{|\bx+\gamma-\bx'|^3}=\frac{1}{|x+\gamma-x'|^\delta}\frac{1}{|x+\gamma-x'|^{3-\delta}}\leqslant\frac{1}{(1-\alpha)^\delta}\frac{1}{(|\gamma|-\sqrt{2})^{3-\delta}}\]
In order to bound the first term, we see that there exists a constant $C$ such that for all $|\alpha|\leq\frac{1}{2}$, $\bx\in\Omega_\alpha$, $\bx'\in\Omega$ and $\bgamma\in\Z^2\setminus\{0\}$, we have
\[\frac{1}{|\bx+\bgamma-\bx'|^3}\leq \frac{C}{|\bgamma|^3}.\]

Thus the first term is bounded by \[C\|\bUpsilon\|_{L^1}\|\bPsi\|_{L^1}.\]

 %The rest of the proof is similar to the one of Lemma~2.4 of \cite{barbaroux}.
 
For the second term, we have to bound
\[
 \iint_{\Omega_\alpha\times\Omega}
 |\bUpsilon(\bx)|\frac{1}{|\bx-\bx'|^{2-d'}}| \bPsi(\bx')|d\bx d\bx'.
\]
 
Hardy-Littlewood-Sobolev inequality (cf. \cite[Theorem~4.3]{liebloss}) gives that there exists $C$ such that:
 \[\iint_{\Omega_\alpha\times\Omega}|\bUpsilon(\bx)|\frac{1}{|\bx-\bx'|^{2-d'}}|\bPsi(\bx')|d\bx d\bx'\leqslant C\|\bUpsilon\|_\frac{4}{2+d'}\|\bPsi\|_\frac{4}{2+d'}.\]
 
%Let $\bUpsilon=|\bchi_\alpha|^{1/2}\bbf$. 
By H\"older's inequality, 
\begin{align*}
 \||\bchi_{\alpha}|^{1/2}\bbf\|_\frac{4}{2+d'}
 &=\left(\int_\Omega|\bchi_{\alpha}|^\frac{2}{2+d'}
 |\bbf|^\frac{4}{2+d'}\right)^\frac{2+d'}{4}\\
 &\leqslant\left(\||\bchi_{\alpha}|^\frac{2}{2+d'}\|_\frac{2+d'}{d'}
 \||\bbf|^\frac{4}{2+d'}\|_\frac{2+d'}{2}\right)^\frac{2+d'}{4}\\
 & =\||\bchi_{\alpha}|\|_\frac{2}{d'}^{1/2}\|\bbf\|_2.    
\end{align*}
A simple change of variable gives us that $\||\bchi_{\alpha}|\|_\frac{2}{d'}=\alpha^{d'}\||\bchi|\|_\frac{2}{d'}\leq\alpha^{d'}\||\bchi|\|_\infty$.    
 
Hence, picking $\bUpsilon=\bPsi=|\bchi_\alpha|^{1/2}\bbf$ in the above estimates yields 
\[
 \left|\langle \bbf, |\bchi_\alpha|^{1/2}(h_\bk^{(0)}\pm i)^{-1}|
 \bchi_\alpha|^{1/2}\bbf\rangle\right|
 \leq C\|\bchi_\alpha \bbf\|_{1}^2+C\alpha^{d'}\| \bbf\|_{2}^2.
\]
 
An application of Cauchy-Schwarz inequality gives that 
\[
 \||\bchi_\alpha| \bbf\|_{1}
 \leqslant \||\bchi_{\alpha}|\|_{L^2}\|\bbf\|_ {2}
 \leq \sqrt{3}\alpha\|\bbf\|_ {2}.
\]
 
This concludes the proof of \eqref{eq:res-est-1}.

The proof of \eqref{eq:res-est-2} is similar: we have to take $\bUpsilon=|\bchi_\alpha|\bbf$ and $\bPsi=\bbf$. We do not give further details.

If $d=2$, we can prove these inequalities for any $\tilde{d}<2$ and then take the supremum.
\end{proof}
%%%%%%%%%%%%%%%%%%%%%%%%%%%%%%%%%%%%%%%%%%%%%%%%%%%

%%%%%%%%%%%%%%%%%%%%%%%%%%%%%%%%%%%%%%%%%%%%%%%%%%%
\begin{lemme}\label{lem4}
 For $\bbf\in Q_\bz\Dom(h_\bk^{(0)})$, we have
 $\|h_\bk^{(0)}Q_\bz \bbf\|\geqslant \pi^d K_0'\|\bbf\|$.
\end{lemme}
%%%%%%%%%%%%%%%%%%%%%%%%%%%%%%%%%%%%%%%%%%%%%%%%%%%

%%%%%%%%%%%%%%%%%%%%%%%%%%%%%%%%%%%%%%%%%%%%%%%%%%%
\begin{proof}
 As we show in Proposition~\ref{appbloch} in the appendix, we have, for $m\in\mathbb{Z}^2$ and $k\in\Omega$, 
\[
 P_\bm h_\bk^{(0)}P_\bm=\bs\cdot \bF(2\pi(\bm-\bk))P_\bm.
\]
 
Hence, according to inequality~\eqref{F} of Hypothesis~\ref{hyp1}(ii), we have for $\bbf\in Q_\bz\Dom(h_\bk^{(0)})$ :
\begin{align*}
 \|h_\bk^{(0)}Q_\bz \bbf\|^2
 =\sum_{\bm\neq0}\|\sigma\cdot \bF(2\pi(\bm-\bk))P_\bm \bbf\|^2&
 \geqslant K_0'^2\sum_{m\neq0}(2\pi|\bm -\bk|)^{2d}\|P_\bm \bbf\|^2\\
 &\geqslant K_0'^2\pi^{2d}\sum_{\bm\neq 0}\|P_\bm \bbf\|^2.
 \end{align*}
\end{proof}
%%%%%%%%%%%%%%%%%%%%%%%%%%%%%%%%%%%%%%%%%%%%%%%%%%%

% If $d\geqslant 1$, the rest of the proof is identical to the one in~\cite{barbaroux} as the shape of the operators $H_0$ and $h_\bk^{(0)}$ is not used anymore. To the contrary, if $d<1$, some  modifications are needed.
% From now on, we denote abusively \[d=\min(d,1).\]
For a self-adjoint operator $T$ and an orthogonal projection $Q$, we define the resolvent set
\begin{align*}
  \rho_{Q}(T):=\left\{ z\in\C\,\mbox{ such that }\,Q(T-z)Q:
{\rm Ran}Q\to  {\rm Ran}
Q\quad \mbox{is invertible}\right\}.
\end{align*} 
We set 
\begin{align*}
  &R_0(z):=\left(Q_\bz(h_{\bf k}^{(0)}-z)Q_\bz\!\!\upharpoonright_{{\rm
  Ran}Q_\bz} \right)^{-1}, \quad z\in  \rho_{Q_\bz}(h_{\bf k}^{(0)}),\\
&R(z):=\left(Q_\bz(h_{\bf k}(\alpha,\beta)-z)Q_\bz\!\!\upharpoonright_{{\rm
  Ran}Q_\bz} \right)^{-1}, \quad z\in  \rho_{Q_\bz}(h_{\bf k}(\alpha,\beta)).
\end{align*}

For $i\in\{1,2,3\}$, we define the  operators $U_i:L^2(\Omega,\C^2) \to \mathrm{Ran} Q_\bz$ and $W_i: \mathrm{Ran}Q_\bz \to L^2(\Omega,\C^2)$ by:
\[W_i=\sqrt{\beta}(\sqrt{|\chi_{i,\alpha}|}\sigma_i)Q_\bz\text{  and  } U_i=\sqrt{\beta}Q_\bz\text{sgn}(\chi_{i,\alpha})\sqrt{|\chi_{i,\alpha}|}.\]

%We will denote \[WR_0(z)U=\sum_{i=1}^3 W_iR_0(z)U_i.\]

%%%%%%%%%%%%%%%%%%%%%%%%%%%%%%%%%%%%%%%%%%%%%%%%%%%
\begin{lemme}\label{lem3.2}
 There exists $C>0$ independent of $\alpha\in ]0,1/2]$ and 
 $\beta>0$ such that for any $|z|\leqslant K_0'\pi^d/2$ and any $j,l$
 \[\|W_jR_0(z)U_l\|\leqslant C\alpha^{d'}\beta.\]
\end{lemme}
%%%%%%%%%%%%%%%%%%%%%%%%%%%%%%%%%%%%%%%%%%%%%%%%%%%

%%%%%%%%%%%%%%%%%%%%%%%%%%%%%%%%%%%%%%%%%%%%%%%%%%%
\begin{proof}
% The proof is almost the same as the one of Lemma~2.6 of \cite{barbaroux}.
Due to Lemma~\ref{lem4}, we have for $|z|\leqslant K_0'\pi^d/2$,
\begin{equation}\label{something}
 \|R_0(z)\|\leqslant\frac{2}{K_0'\pi^d} \, .
\end{equation}
Using the first resolvent identity, we get for $j,l=1,2,3$:
\begin{equation}\label{eq:re-1}
\begin{split}
 W_j R_0(z) U_l = W_j R_0(i) U_l + (z-i) W_j R_0(i)^2 U_l \\+ (z-i)^2 W_j R_0(i)R_0(z)R_0(i)  U_l .
\end{split}
\end{equation}
We shall separately estimate each term on the right hand side of \eqref{eq:re-1}.\\
Since $h_\bk^{(0)}$ commutes with the projections $P_\bm$ we have
\begin{equation}\label{eq:re-4}
 R_0(i) =    
 (h_\bk^{(0)} - i )^{-1}   
 -  (P_\bz  (h_\bk^{(0)} - i )P_\bz\!\!\upharpoonright_{{\rm Ran}P_\bz} )^{-1} .
\end{equation}

Note that due to the definition of $U_l$ and $W_j$ for any $z$ such that $|z|\leqslant K_0'\pi^d/2$
\begin{equation}\label{eq:re-added}W_j(P_\bz  (h_\bk^{(0)} - z )P_\bz\!\!\upharpoonright_{{\rm Ran}P_\bz} )^{-1}U_l=0.\end{equation} 
% 
% Hence 
% 
%  \| W_j (P_\bz (h_\bk^{(0)} -i)P_\bz\!\!\upharpoonright_{{\rm Ran}P_\bz})^{-1}U_l\| =0
% \end{equation}
% 

The identity \eqref{eq:re-4} together with  inequality
\eqref{eq:res-est-1} and \eqref{eq:re-added} imply that there exists $c>0$, such that  for $|\alpha|<1/2$ and all $k\in\Omega$
\begin{equation}\label{eq:re-8}
\begin{split}
\| W_j R_0(i)U_l \| & = \| W_j (h_\bk^{(0)} -i)^{-1}U_l \| \\
& \leq c  \beta\alpha^{d'}.
\end{split}
\end{equation}
This bounds the first term on the right hand side of \eqref{eq:re-1}.

To estimate the second one we first notice that
\begin{equation*}
\begin{split}
  (h_\bk^{(0)} - i)^{-2}   =  
  (Q_\bz  (h_\bk^{(0)} - i )Q_\bz\!\!\upharpoonright_{{\rm Ran}Q_\bz})^{-2} 
  + (P_\bz  (h_\bk^{(0)} - i )P_\bz\!\!\upharpoonright_{{\rm Ran}P_\bz})^{-2}.
\end{split}
\end{equation*}
It is easy to see that the equation \eqref{eq:re-added}  remains true with a power -2. Then,
\begin{equation}\label{eq:re-9}
\begin{split}
 \| (z-i) &W_j R_0(i)^2 U_l \| = \| (z-i) W_j  (h_\bk^{(0)} - i )^{-2} U_l\| \\
 & \leq \sqrt{1+\frac{K_0'^2\pi^{2d}}{4}}\beta \| \sqrt{|\bchi_\alpha |}(h_\bk^{(0)} - i )^{-1} \| \, \|(h_\bk^{(0)} - i )^{-1} \sqrt{|\bchi_\alpha|} \|\\ % + \pi\beta\alpha^2 \\
 & \leq C \beta \alpha^{d'},
\end{split}
\end{equation}
 for some $C>0$ independent of $\alpha$ and $\beta$, where we used \eqref{eq:res-est-2} in the last inequality.

Finally, we bound the last term on the right hand side of
 \eqref{eq:re-1}. Observe that from  inequalities \eqref{eq:int-1} and \eqref{eq:res-est-2}
we obtain that there exists $c,C>0$ such that for all $\alpha\le 1/2$
\begin{align*}
 \| \sqrt{|\bchi_\alpha|} R_0(i) \|& \leq \| \sqrt{|\bchi_\alpha|} (h_\bk^{(0)}-i)^{-1}\|
 + \| \sqrt{|\bchi_\alpha|} P_\bz (P_\bz (h_\bk - i)P_\bz)^{-1} \| \\&\leq c 
\sqrt{\alpha^{d'}}+c\alpha\leq C\sqrt{\alpha^{d'}}.
\end{align*}
Therefore, using \eqref{something} and \eqref{eq:res-est-2}
\begin{align*}
 &\| (z-i)^2 W_j R_0(i)R_0(z) R_0(i) U_l \| \\ &\leq \beta |z-i|^2
 \| \sqrt{|\bchi_\alpha|} R_0(i) \|\, \| R_0(z)\|\, \|R_0(i)
 \sqrt{|\bchi_\alpha|}\|\\
 & \leq  C \beta |z-i|^2 \alpha^{d'}. 
\end{align*}
Summing the latter bound together with \eqref{eq:re-8} and
\eqref{eq:re-9} (in view of \eqref{eq:re-1}) concludes the proof.
%  We can see that 
%  \[W_i(P_\bz(h_\bk^{(0)}-i)P_\bz)^{-1}U_i=\beta(\sqrt{|\chi_{i,\alpha}|}\sigma_i)Q_\bz(P_\bz(h_\bk^{(0)}-i)P_\bz)^{-1}Q_\bz(\text{sgn}\sqrt{|\chi_{i,\alpha}|})=0\]
%   since $P_\bzQ_\bz=0$. The result is the same with a power $-2$.
%   
%   For the rest, it suffices to replace $\alpha$ by $\alpha^d$ and to sum the 3 inequalities.
\end{proof}
%%%%%%%%%%%%%%%%%%%%%%%%%%%%%%%%%%%%%%%%%%%%%%%%%%%

%Le lemme 2.7 est toujours valable et sa preuve ne change pas.
%%%%%%%%%%%%%%%%%%%%%%%%%%%%%%%%%%%%%%%%%%%%%%%%%%%
\begin{lemme}\label{resolventidentity}
 Let \[\mathcal{S}=\{z\in\rho_{Q_\bz}(h_\bk^{(0)}):\sup_{j,l}\|W_jR_0(z)U_l\|<1/3\}.\]
 Then, for every $z\in\mathcal{S}$, we have $z\in \rho_{Q_\bz}(h_\bk(\alpha,\beta))$.% and 
 %\[R(z)=R_0(z)-R_0(z)U(1+WR_0(z)U)^{-1}WR_0(z).\]
\end{lemme}
%%%%%%%%%%%%%%%%%%%%%%%%%%%%%%%%%%%%%%%%%%%%%%%%%%%

%%%%%%%%%%%%%%%%%%%%%%%%%%%%%%%%%%%%%%%%%%%%%%%%%%%
\begin{proof}
 Let $z\in\mathcal{S}\cap\mathbb{R}$. Put $z_\epsilon=z+i\epsilon$. The set $\mathcal{S}$ being open, for $\epsilon>0$ small enough $z_\epsilon\in\mathcal{S}$.
 We denote $\bU=(U_1,U_2,U_3)$ and $\bW=(W_1,W_2,W_3)$.
 Applying the second resolvent identity several times, we find for any $N>0$:
 \[ R(z_\epsilon)=R_0(z_\epsilon)\left(\sum_{n=0}^N(-1)^n\left(\bU\bW^TR_0(z_\epsilon)\right)^n+T_{N+1}\right)\]
 with \[T_{N+1}=(-1)^{N+1}\left(\bU\bW^TR_0(z_\epsilon)\right)^N\bU\bW^TR(z_\epsilon).\]
%  
%  But, for $n\geqslant 1$, by convexity of the function $x\mapsto x^n$ on $\mathbb{R}^+$, we have 
%  \[\|\left(\sum_i U_iW_iR_0(z_\epsilon)\right)^n\|\leqslant\|\sum_i U_iW_iR_0(z_\epsilon)\|^n\leqslant 3^{n-1}\sum_i\| U_iW_iR_0(z_\epsilon)\|^n.\]
%  Moreover, \[ 3^{N-1}\| U_iW_iR_0(z_\epsilon)\|^n\leqslant3^{N-1}\|  U_i\|\|W_iR_0(z_\epsilon)U_i\|^{N-1}\|W_i\|\]
%  which tends to 0 and $N$ goes to infinity since $\|W_iR_0(z_\epsilon)U_i\|<\frac{1}{3}$.
 But we have that
\begin{align*}
  \left(\bU\bW^TR_0(z_\epsilon)\right)^N
  =\sum_{i_1,...,i_N}U_{i_1}W_{i_1}R_0(z_\epsilon)
  U_{i_2}...W_{i_{N-1}}R_0(z_\epsilon)U_{i_N}W_{i_N}R_0(z_\epsilon)
\end{align*}
and then
\[\left\| \left(\bU\bW^TR_0(z_\epsilon)\right)^N\right\|\leq C 3^N \left(\sup_{j,l}\|W_jR_0(z_\epsilon)U_l\|\right)^{N-1}\]
which tends to 0 as $N\to \infty$ since $z_\epsilon\in \mathcal{S}$.

Then, at fixed $\epsilon$, we have
\begin{equation}\label{eq:reso}
 R(z_\epsilon)
 =R_0(z_\epsilon)\sum_{n=0}^\infty(-1)^n
 \left(\bU\bW^TR_0(z_\epsilon)\right)^n.
\end{equation}
 
Using the definition of $\mathcal{S}$ and equation \eqref{something}, we obtain that this resolvent is bounded uniformly in $\epsilon$, so we can take the limit $\epsilon\to 0$ and thus
\[
 \forall z \in\mathcal{S} ; z\in \rho_{Q_\bz}(h_\bk(\alpha,\beta)).
\]
\end{proof}
%%%%%%%%%%%%%%%%%%%%%%%%%%%%%%%%%%%%%%%%%%%%%%%%%%%

We are now ready to study the invertibility of Feshbach's operator $\mathcal{F}_{P_\bz}(z)$ for $\alpha^{d'}\beta$ small enough.
%We recall that \[\mathcal{F}_{P_\bz}(z)=P_\bz(h_\bk(\alpha,\beta)-z)P_\bz-\beta^2 P_\bz\chi_\alpha\cdot\sigma Q_\bz(Q_\bz(h_\bk(\alpha,\beta)-z)Q_\bz)^{-1}Q_\bz\chi_\alpha\cdot\sigma P_\bz\]
 To this purpose, we use the two following lemmas (similar to Lemmas~2.2 and~2.3 of \cite{barbaroux}):

%%%%%%%%%%%%%%%%%%%%%%%%%%%%%%%%%%%%%%%%%%%%%%%%%%%
\begin{lemme}\label{lem2.2}
There exists a constant $\delta\in]0,1[$ such that, for all $\alpha\in]0,1/2[$ and $\beta>0$ satisfying $\alpha^{d'}\beta<\delta$,
 $Q_\bz(h_\bk(\alpha,\beta)-z)Q_\bz$ is invertible on the range of $Q_\bz$ for all $z\in [-K_0'\pi^d/2,K_0'\pi^d/2]$ and $\bk\in\Omega$.
\end{lemme}
%%%%%%%%%%%%%%%%%%%%%%%%%%%%%%%%%%%%%%%%%%%%%%%%%%%

%%%%%%%%%%%%%%%%%%%%%%%%%%%%%%%%%%%%%%%%%%%%%%%%%%%
\begin{proof}%[Proof of Lemmas \ref{lem2.2} and \ref{lem2.3}]
Notice that the proof of this lemma follows from
Lemma~\ref{resolventidentity} since  $z\in \mathcal{S}$ provided
$\alpha^{d'}\beta$ is small enough, according to Lemma~\ref{lem3.2}. 
\end{proof}
%%%%%%%%%%%%%%%%%%%%%%%%%%%%%%%%%%%%%%%%%%%%%%%%%%%

Put 
\[
 \mathcal{B}_{P_\bz}(z)
 =\beta^2 P_\bz\bchi_\alpha\cdot\bs 
 Q_\bz(Q_\bz(h_\bk(\alpha,\beta)-z)Q_\bz)^{-1}Q_\bz\bchi_\alpha\cdot\bs P_\bz.
\]

%%%%%%%%%%%%%%%%%%%%%%%%%%%%%%%%%%%%%%%%%%%%%%%%%%%
\begin{lemme}\label{lem2.3}
There exist two constants $\delta\in]0,1[$ and $C>0$ such that, for all $\alpha\in]0,1/2[$ and $\beta>0$ satisfying $\alpha^{d'}\beta<\delta$, we have
\[\|\mathcal{B}_{P_\bz}(z)\bpsi\|\leqslant C\beta^2\alpha^{2+d'}\|\bpsi\|\] for all $z\in [-K_0'\pi^d/2,K_0'\pi^d/2]$, $k\in\Omega$ and $\bpsi \in L^2(\Omega,\mathbb{C}^2)$.
\end{lemme}
%%%%%%%%%%%%%%%%%%%%%%%%%%%%%%%%%%%%%%%%%%%%%%%%%%%

\begin{proof}%[Proof of Lemmas \ref{lem2.2} and \ref{lem2.3}]
 According to equation~\eqref{eq:reso}, we have that
\begin{align*}
 R(z)&=R_0(z)+R_0(z)\sum_{n\geq1}(-1)^n\left(\bU\bW^TR_0(z)\right)^n\\
 &=R_0(z)-R_0(z)\sum_{n\geq0}(-1)^n\bU\left(\bW^TR_0(z)\bU\right)^n\bW^TR_0(z)\\
 &=R_0(z)-R_0(z)\bU(I_3+\bW^TR_0(z)\bU)^{-1}\bW^T R_0(z).
\end{align*}
We remark that $\bW^TR_0(z)\bU$ is an operator acting on $\left(L^2(\Omega, \C^2)\right)^3$.

The definition of $\mathcal{B}_{P_\bz}$ and these equalities give us 2 terms to estimate.
On the one hand,
\begin{equation*}
\begin{split}
  \label{eq:3.23}  
  \beta^2\|P_\bz\bchi_\alpha\cdot\bs R_0(z)
  \bchi_\alpha\cdot\bs P_\bz \|&
  \le \beta\|P_\bz |\bchi_\alpha|^{1/2}\| \sum_{i,j}\|W_iR_0(z) U_j\| 
  \||\bchi_\alpha|^{1/2} P_\bz\|\\&
  \le c\beta^2\alpha^{2+d'},
\end{split}
\end{equation*}
where we used Lemma \ref{lem3.2} and \eqref{eq:int-1}. 

On the other hand,
assuming that $\alpha^{d'}\beta$ is so small that $ \|\bW^TR_0(z)\bU\|<1/2$ we have
\begin{align*}
   &\beta^2\|P_\bz \bchi_\alpha \cdot \bs
   R_0(z)  \bU(I_3+\bW^TR_0(z)\bU)^{-1}
   \bW^T R_0(z)\bchi_\alpha\cdot\bs P_\bz\|\\
   &\le \beta \|P_\bz|\bchi_\alpha|^{1/2}\|  
   \sum_{i,j}\|W_iR_0(z) U_j\|
   \| (I_3+\bW^TR_0(z)\bU)^{-1}\| \\ 
   &\ \ \ \times \sum_{i,j}\|W_iR_0(z) U_j\|
   \||\bchi_\alpha|^{1/2} P_\bz\|
   \le c \beta^3\alpha^{2+2d'}.
\end{align*}
The latter inequality together with \eqref{eq:3.23} finishes the proof
of the lemma.
\end{proof}
\begin{proof}[Proof of Theorem~\ref{thm:main}]
  In view of \eqref{spectra} it is enough to show the invertibility of
  the Feshbach operator uniformly in $\bk\in\Omega$. Using Lemmas %here
  \ref{lem2.1} and \ref{lem2.3} we get that for any
  $\bpsi\in P_\bz L^2(\Omega,\C^2)$
\begin{align*}
  \norm{\fesh_{P_\bz}(z)\bpsi}&\ge  \norm{ (P_\bz(h_\bk(\alpha,\beta) - z)P_\bz ) \bpsi} -\norm {\B _{P_\bz}(z)\bpsi}\\
         &\ge (\beta\alpha^2 \frac{|\bP|}{2}-|z|-c\alpha^{2+d'} \beta^2)\norm{\bpsi}.
\end{align*}
This concludes the proof by picking $\alpha^{d'}\beta$ so small that
$\frac{|\bP|}{2} > c\alpha^{d'}\beta$. 
The theorem is then proven with $C=\frac{c}{2}$.
\end{proof}
%%%%%%%%%%%%%%%%%%%%%%%%%%%%%%%%%%%%%%%%%%%%%%%%%%%

%%%%%%%%%%%%%%%%%%%%%%%%%%%%%%%%%%%%%%%%%%%%%%%%%%%
%%%%%%%%%%%%%%%%%     APPENDIX    %%%%%%%%%%%%%%%%%
%%%%%%%%%%%%%%%%%%%%%%%%%%%%%%%%%%%%%%%%%%%%%%%%%%%
\appendix
\section{Bloch-Floquet transformation}
In this appendix, we study the Bloch-Floquet transformation applied to our operator $H(\alpha,\beta)$. Because the potential is bounded and $\Z^2$-periodic, it is enough to study this transformation applied to $H_0$. Because $H_0$ is unbounded, we prefer to work with its resolvent and we start with the following proposition.
Recall that we have denoted by $(H_0-i)^{-1}(\bx,\bx')$ the integral kernel of $(H_0-i)^{-1}$.

%%%%%%%%%%%%%%%%%%%%%%%%%%%%%%%%%%%%%%%%%%%%%%%%%%%
\begin{prop}\label{appnoyint}
Let $\mathcal{U}$ be the Bloch-Floquet transformation as defined in \eqref{defbloch}.
\[
 \mathcal{U}(\bs\cdot \bF(-i\nabla)-i)^{-1}\mathcal{U}^*
 =\intdir{\Omega}{d\bk}{g_{\bk}}\, ,
\] 
where, for $\bk\in\Omega$, the operator $g_\bk:L^2(\Omega,\C^2)\to L^2(\Omega,\C^2)$ has an integral kernel given by 
\[
 g_\bk(\bx,\bx')=\sum_{\bgamma\in\Z^2}e^{2i\pi\bgamma\cdot \bk}e^{2i\pi\bx\cdot \bk}(H_0-i)^{-1}(\bx+\bgamma,\bx')e^{-2i\pi\bx'\cdot \bk}.
\]
\end{prop}
%%%%%%%%%%%%%%%%%%%%%%%%%%%%%%%%%%%%%%%%%%%%%%%%%%%

%%%%%%%%%%%%%%%%%%%%%%%%%%%%%%%%%%%%%%%%%%%%%%%%%%%
\begin{proof}
%Lemma~\ref{reso} gives us that the integral kernel $K_0\in L^1(\mathbb{R}^2\times\R^2)$ so it is possible to interchange sums and integrals in the rest of the calculation.
Let $\bbf\in\mathcal{C}^\infty(\Omega\times\R^2)^2$ such that $\bbf$ is $\Z^2$-periodic with respect to the second variable and $\bx\in\mathbb{R}^2$.
To avoid heavy notation, we will denote $K_0(\bx,\bx')=(H_0-i)^{-1}(\bx,\bx')$. We can then write
\begin{align*}
 ((\bs\cdot \bF(-i\nabla&)-i)^{-1}\mathcal{U}^*\bbf)(\bx)=\int_{\mathbb{R}^2}K_0(\bx,\bx')(\mathcal{U}^*\bbf)(\bx')d\bx'\\
 =&\sum_{\bgamma'\in\mathbb{Z}^2}\int_{\Omega}K_0(\bx,\bx'+\bgamma')\int_\Omega e^{-2i\pi(\bx'+\bgamma')\cdot \bk'} \bbf(\bx',\bk')d\bk'd\bx'.
\end{align*}

The Fourier coefficients  $\hat{\bbf}(\bx',\bgamma')=\int_\Omega e^{-2i\pi(\bx'+\bgamma')\cdot \bk'} \bbf(\bx',\bk')d\bk'$ decay faster than any polynomial in $\bgamma'$ uniformly in $\bx'$. The integral kernel $K_0(\bx+\bgamma, \bx'+\bgamma')$ has a decay like $|\bgamma-\bgamma'|^{-3}$ when $|\bgamma-\bgamma'|$ is larger than 3 uniformly in $\bx$ and $\bx'$. Moreover, $K_0(\bx+\bgamma, \bx'+\bgamma')$ is absolutely integrable with respect to $\bx'$ and 
\[
 \int_{\Omega} |K_0(\bx+\bgamma, \bx'+\bgamma')|d\bx'\leq C
\]
uniformly in $\bx\in\Omega$ and $|\bgamma-\bgamma'|\leq 3$. These facts justify the interchange of the various series below:
\begin{align*}
 &(\mathcal{U}(\bs\cdot \bF(-i\nabla)-i)^{-1}\mathcal{U}^*\bbf)(\bx,\bk)\\
  =&\sum_{\bgamma\in\mathbb{Z}^2}e^{2i\pi \bk\cdot(\bx+\bgamma)}\sum_{\bgamma'\in\mathbb{Z}^2}\int_{\Omega}K_0(\bx+\bgamma,\bx'+\bgamma')\hat{\bbf}(\bx',\bgamma')d\bx'\\
  =&\sum_{\bgamma'\in\mathbb{Z}^2}\sum_{\bgamma\in\mathbb{Z}^2}e^{2i\pi \bk\cdot \bx}e^{2i\pi \bk\cdot(\bgamma-\bgamma')}\int_{\Omega}K_0(\bx+(\bgamma-\bgamma'),\bx') e^{2i\pi \bgamma'\cdot \bk}\hat{\bbf}(\bx',\bgamma')d\bx'\\
  =&\sum_{\bgamma'\in\mathbb{Z}^2}\sum_{\tilde{\bgamma}\in\mathbb{Z}^2}e^{2i\pi \bk\cdot\tilde{\bgamma}}e^{2i\pi \bk\cdot \bx}\int_{\Omega}K_0(\bx+\tilde{\bgamma},\bx')e^{2i\pi \bk\cdot\bgamma'}\hat{\bbf}(\bx',\bgamma')d\bx'\\
  =&\int_{\Omega}\sum_{\bgamma\in\mathbb{Z}^2}e^{2i\pi \bk\cdot (\bx-\bx')}e^{2i\pi \bk\cdot\bgamma}K_0(\bx+\bgamma,\bx')\sum_{\bgamma'\in\mathbb{Z}^2 }e^{2i\pi \bk\cdot(\bgamma'+\bx')}\hat{\bbf}(\bx',\bgamma')d\bx'.
\end{align*}

In the last line, we identify the Fourier series representation of $\bbf(\bx',\cdot)$ at the point $\bk$. We finally obtain 
$$
\int_{\Omega}\sum_{\bgamma\in\mathbb{Z}^2}e^{2i\pi \bk\cdot(\bx+\bgamma-\bx')}K_0(\bx+\bgamma,\bx') \bbf(\bx',\bk)d\bx'=\int_\Omega g_{k}(\bx,\bx')\bbf(\bx',\bk)d\bx',
$$
which concludes the proof of Proposition~\ref{appnoyint}.
\end{proof}
%%%%%%%%%%%%%%%%%%%%%%%%%%%%%%%%%%%%%%%%%%%%%%%%%%%

%%%%%%%%%%%%%%%%%%%%%%%%%%%%%%%%%%%%%%%%%%%%%%%%%%%
\begin{prop}\label{appbloch}
We have
 \begin{equation*}
  \mathcal{U}H_0\mathcal{U}^*=\intdir{\Omega}{d\bk}{h_\bk^{(0)}}\, ,
 \end{equation*}
 with
 \begin{equation*}
  h_\bk^{(0)}=\bs\cdot \bF(-i\nabla_{\mathrm{per}}-2\pi \bk).
 \end{equation*}
\end{prop}
%%%%%%%%%%%%%%%%%%%%%%%%%%%%%%%%%%%%%%%%%%%%%%%%%%%

%%%%%%%%%%%%%%%%%%%%%%%%%%%%%%%%%%%%%%%%%%%%%%%%%%%
\begin{proof}
 
According to Theorem~XIII.85 of \cite{RS4}, to prove \eqref{bloch}, we need
to show that, for $\bk\in\Omega$  $g_\bk:L^2(\Omega,\C^2)\to L^2(\Omega,\C^2)$ satisfies \[g_{\bk}=(h_\bk^{(0)}-i)^{-1}.\]

To this purpose, we will denote by $(e_j)_{j=1,2}$ the vectors of the standard basis in $\C^2$ and $\Psi_\bm=e^{2i\pi \bm\cdot \bx}$. 
We will prove that for all $\bm\in\Z^2$ and $j\in\{1,2\}$ we have \[g_{\bk}(\Psi_\bm\otimes e_j)=(h_\bk^{(0)}-i)^{-1}(\Psi_\bm\otimes e_j)=(\bs\cdot \bF(2\pi(\bm-\bk))-i)^{-1}(\Psi_\bm\otimes e_j).\]

Recall the notation $G(\bp)=(\bs\cdot \bF(\bp)-i)^{-1}\in\mathcal{B}(\C^2)$.
We have that
\begin{equation}\label{hc2}
\begin{split}
\{ & g_{k}(\Psi_{\bm}\otimes e_j)\}(\bx)\\
& =\int_{\Omega}g_{k}(\bx,\bx')(\Psi_{\bm}\otimes e_j)(\bx')d\bx' \\
& =\Psi_{\bm}(\bx) \int_{\Omega}d\bx'\sum_{\bgamma\in\mathbb{Z}^2}e^{2i\pi \bk\cdot(\bx+\bgamma-\bx')}e^{2i\pi \bm\cdot(-\bx+\bx')}\frac{1}{2\pi} \mathcal{F}^{-1}( G )(\bx-\bx'+\bgamma)e_j.
\end{split}
\end{equation}
Because both $\bm$ and $\bgamma$ are in $\mathbb{Z}^2$ we have $e^{2\pi i \bm\cdot \bgamma}=1$, hence in \eqref{hc2} we can replace 
 $e^{2i\pi \bm\cdot(-\bx+\bx')}$ with $e^{-2i\pi \bm\cdot(\bx+\bgamma-\bx')}$. Thus, after a change of variables we obtain 
\begin{equation*}%\label{hc3}
\begin{split}
\{g_{\bk}(\Psi_{\bm}\otimes e_j)\}(\bx)&=\Psi_{\bm}(\bx)\frac{1}{2\pi} \int_{\mathbb{R}^2} e^{-2i\pi (\bm-\bk)\cdot \by}\mathcal{F}^{-1}( G )(\by)d\by\, e_j \\
&= G(2\pi(\bm-\bk))\; (\Psi_{\bm}\otimes e_j)(\bx).
\end{split}
\end{equation*}
This ends the proof of the proposition.
\end{proof}
%%%%%%%%%%%%%%%%%%%%%%%%%%%%%%%%%%%%%%%%%%%%%%%%%%%

\noindent\textbf{Acknowledgments.} It is a pleasure to thank the REB program of CIRM for giving us the opportunity to start this research. JMB and SZ thank the Department of Mathematical Sciences of Aalborg Universitet for its hospitality where part of this work was done. JMB and HC thank CRM Montr\'eal for its hospitality.

%%%%%%%%%%%%%%%%%%%%%%%%%%%%%%%%%%%%%%%%%%%%%%%%%%%
%%%%%%%%%%%%%%%  BIBLIOGRAPHY     %%%%%%%%%%%%%%%%%
%%%%%%%%%%%%%%%%%%%%%%%%%%%%%%%%%%%%%%%%%%%%%%%%%%%
%\bibliographystyle{plain}
%\bibliography{graphene}

\begin{thebibliography}{1}

\bibitem{barbaroux}
J.-M. Barbaroux, H.D. Cornean, and E.~Stockmeyer,
\newblock \textit{Spectral gaps in graphene antidot lattices}.
\newblock Integral Equations and Operator Theory \textbf{89} (2017), 631--646.
\newblock \href {http://dx.doi.org/10.1007/s00020-017-2411-9}
  {\path{doi:10.1007/s00020-017-2411-9}}.

\bibitem{BCLS} 
J.-M. Barbaroux, H.D. Cornean, L. Le~Treust, and E.~Stockmeyer,
\newblock \textit{Resolvent Convergence to {D}irac Operators on Planar Domains}.
\newblock Ann. Henri Poincar\'e \textbf{20} (2019), 1877--1891. 
\newblock \href {https://doi.org/10.1007/s00023-019-00787-2}
  {\path{doi:10.1007/s00023-019-00787-2}}.

\bibitem{BCZ}
J.-M. Barbaroux, H.D. Cornean, and S.~Zalczer,
\newblock \textit{Localization for gapped Dirac Hamiltionians with random perturbations: Application to graphene antidot lattices}.
\newblock Documenta Mathematica \textbf{24} (2019), 65--93.
\newblock \href {http://dx.doi.org/10.25537/dm.2019v24.65-93}
  {\path{doi:10.25537/dm.2019v24.65-93}}.

\bibitem{CGPNGG}
A.H. Castro~Neto, F.~Guinea, N.M.R. Peres, K.S. Novoselov, and A.K. Geim,
\newblock \textit{The electronic properties of graphene}.
\newblock Rev. Mod. Phys. \textbf{81} (2009), 109--162.
\newblock \href{https://link.aps.org/doi/10.1103/RevModPhys.81.109}
 {\path{doi = {10.25537/dm.2019v24.65-93}}}.


\bibitem{FPFMBPJ}
J.A. F\"urst, J.G. Pedersen, C.~Flindt, N.A. Mortensen, M.~Brandbyge, T.G.
  Pedersen, and A.-P. Jauho,
\newblock \textit{Electronic properties of graphene antidot lattices}.
\newblock New J. Phys. \textbf{11} (2009), 095020 (17pp).
\newblock \href{https://doi.org/10.1088/1367-2630/11/9/095020}
 {\path{doi = {10.1088/1367-2630/11/9/095020}}}.
	

\bibitem{liebloss}
E.H. Lieb and M.~Loss,
\newblock {\em Analysis}.
\newblock CRM Proceedings \& Lecture Notes. American Mathematical Society,
  2001.

\bibitem{MCK}
Edward McCann and Mikito Koshino,
\newblock \textit{The electronic properties of bilayer graphene}.
\newblock Rep. Prog. Phys. \textbf{76} (2013) 056503.
\newblock \href{https://doi.org/10.1088/0034-4885/76/5/056503}
	{\path{doi = {10.1088/0034-4885/76/5/056503}}}.

\bibitem{bloch}
T.~Muthukumar,
\newblock \textit{{B}loch-{F}loquet transform}.
\newblock ATM Workshop on PDE and Fourier Analysis, Uttarpradesh,
  2014.

\bibitem{nenciu}
G.~Nenciu,
\newblock \textit{Dynamics of band electrons in electric and magnetic fields: rigorous justification of the effective {H}amiltonians}.
\newblock Rev. Mod. Phys. \textbf{63} (1991), 91--127.
\newblock  \href{https://link.aps.org/doi/10.1103/RevModPhys.63.91}
 {\path{doi = {10.1103/RevModPhys.63.91}}}.

\bibitem{RS4}
M.~Reed and B.~Simon,
\newblock {\em Methods of Modern Mathematical Physics, volume IV : Analysis of operators}.
\newblock Academic Press, New York, 1978.

\bibitem{Thaller}
B.~Thaller,
\newblock {\em The {D}irac Equation}.
\newblock Springer-Verlag, Berlin, 1992.
\end{thebibliography}

\bigskip

\end{document}